\documentclass[prd,aps,eqsecnum,superscriptaddress,nofootinbib,notitlepage,longbibliography]{revtex4-2}

\usepackage[a4paper, left=2.2cm, right=2.2cm, top=2.5cm, bottom=2.5cm]{geometry}
\usepackage{setspace}
\usepackage[english]{babel}
\addto\extrasenglish{%
}
\usepackage[utf8]{inputenc}
\usepackage[T1]{fontenc}
\usepackage{lmodern}
\usepackage{amsmath,amsthm,mathdots,amssymb,bm,bbm,mathrsfs,mathtools}
\usepackage[final]{graphicx}
\usepackage{subcaption}
\usepackage{tikz,pgfplots}\pgfplotsset{compat=1.16}

\usepackage{verbatim}

\usepackage{thm-restate}
\newtheorem{thm}{Theorem}[section]
\newtheorem{cor}{Corollary}[section]
\newtheorem{lem}{Lemma}[section]

\newtheorem{rmk}{Remark}[section]

\newtheorem{defn}{Definition}[section]

\usepackage{comment}

\usepackage[colorlinks=true, urlcolor=violet, linkcolor=blue, citecolor=red, hyperindex=true, linktocpage=true, pagebackref=true, draft=false, hypertexnames=false]{hyperref}
\usepackage{mleftright}\mleftright

\usepackage{enumitem}
\mathtoolsset{showonlyrefs=true}
\usepackage{microtype,color,graphicx}
\usepackage{tikz-cd}
\usetikzlibrary{shapes.callouts, arrows.meta, positioning, calc}

\usepackage{xcolor}
\usepackage{multirow}
\usepackage{tablefootnote}

\newcommand{\CE}{\mathcal{E}}
\newcommand{\BE}{\mathbb{E}}

\newcommand{\CK}{\mathcal{K}}
\newcommand{\CL}{\mathcal{L}}
\newcommand{\CM}{\mathcal{M}}
\newcommand{\CN}{\mathcal{N}}

\newcommand{\CR}{\mathcal{R}}

\newcommand{\CT}{\mathcal{T}}

\newcommand{\vA}{\bm{A}}

\newcommand{\vB}{\bm{B}}

\newcommand{\vC}{\bm{C}}

\newcommand{\vH}{\bm{H}}
\newcommand{\vI}{\bm{I}}

\newcommand{\vK}{\bm{K}}

\newcommand{\vO}{\bm{O}}
\newcommand{\vP}{\bm{P}}

\newcommand{\vQ}{\bm{Q}}

\newcommand{\vV}{\bm{V}}

\newcommand{\vX}{\bm{X}}
\newcommand{\vY }{\bm{Y }}
\newcommand{\vZ}{\bm{Z}}

\newcommand{\vsigma}{\bm{ \sigma}}

\newcommand{\vrho}{\bm{ \rho}}
\renewcommand{\L}{\left}
\newcommand{\R}{\right}

\newcommand{\dagg}{\dagger}

\newcommand{\vertiii}[1]{{\left\vert\kern-0.25ex\left\vert\kern-0.25ex\left\vert #1 \right\vert\kern-0.25ex\right\vert\kern-0.25ex\right\vert}}
\newcommand{\norm}[1]{\Vert {#1} \Vert}

\newcommand{\normp}[2]{\norm{#1}_{#2}}
\newcommand{\lnormp}[2]{\lnorm{#1}_{#2}}

\newcommand{\labs}[1]{\left\vert {#1} \right\vert}
\newcommand{\lnorm}[1]{\left\Vert {#1} \right\Vert}
\newcommand{\e}{\mathrm{e}}

\newcommand{\ri}{\mathrm{i}}
\newcommand{\rd}{\mathrm{d}}
\newcommand*{\tr}{\mathrm{Tr}}

\newcommand{\indicator}{\mathbbm{1}}

\newcommand{\undersetbrace}[2]{ \underset{#1}{\underbrace{#2}}}

\DeclarePairedDelimiterX{\braket}[1]{\langle}{\rangle}{#1}
\DeclarePairedDelimiterX\ketbra[2]{| }{|}{#1 \delimsize\rangle\!\delimsize\langle #2}
\DeclarePairedDelimiterX\dotp[2]{\langle}{\rangle}{#1, #2}

\DeclareMathAlphabet{\dutchcal}{U}{dutchcal}{m}{n}
\SetMathAlphabet{\dutchcal}{bold}{U}{dutchcal}{b}{n}
\DeclareMathAlphabet{\dutchbcal} {U}{dutchcal}{b}{n}

\makeatletter
\DeclareRobustCommand*{\pmzerodot}{%
	\nfss@text{%
		\sbox0{$\vcenter{}$}
		\sbox2{0}%
		\sbox4{0\/}%
		\ooalign{%
			0\cr
			\hidewidth
			\kern\dimexpr\wd4-\wd2\relax
			\raise\dimexpr(\ht2-\dp2)/2-\ht0\relax\hbox{%
				\if b\expandafter\@car\f@series\@nil\relax
				\mathversion{bold}%
				\fi
				$\cdot\m@th$%
			}%
			\hidewidth
			\cr
			\vphantom{0}
		}%
	}%
}

\usetikzlibrary{quantikz2}
\makeatletter
\def\l@subsubsection#1#2{}
\makeatother
\begin{document}

\title{Note on Strong Quantum Markov Properties}

 \author{Chi-Fang Chen}
 \email{achifchen@gmail.com}
\affiliation{University of California, Berkeley, CA, USA}
\affiliation{Massachusetts Institute of Technology, Cambridge, MA, USA}

\begin{abstract}

Quantum many-body Gibbs states satisfy an approximate local Markov property~\cite{chen2025GibbsMarkov}: local noise can be approximately recovered by a quasi-local recovery map, and the conditional mutual information decays for the corresponding tripartition. Recent work~\cite{bergamaschi2025structural} extends this property to approximate stationary states (metastable states) of certain master equations modeling system--bath dynamics, and proposes a strengthened post-selected recovery property—requiring recovery to hold for each measurement outcome—motivated in part by potential applications to quantum simulation. In this note, we characterize this \textit{strong Markov property}: it holds if and only if the state additionally satisfies correlation decay for suitable pairs of observables.

We further prove several structural and operational consequences of the strong Markov property in the presence of an underlying master equation. First, one can estimate multiple observables from a \textit{single copy} of the state via a repeated measurement--recovery protocol. Second, any two strongly Markov states must have local marginals that are either very close or well separated. Third, if a strongly Markov state can be expressed as a mixture of two strongly Markov states, then their local marginals must be nearly indistinguishable.

\end{abstract}
\maketitle

%\tableofcontents

\section{Introduction}
The Markov property is a fundamental structural feature of classical Gibbs distributions: any subset of variables, conditioned on a suitable boundary, is independent of the rest of the system. Many analytic and algorithmic results, explicitly or implicitly, rely on such conditional independence.

For quantum Gibbs states, however, an exact analog fails~\cite{kuwahara2024clustering, bakshi2025dobrushin, kato2025clustering}. Fortunately, an approximate local Markov property still holds generally for quantum Gibbs states~\cite{chen2025GibbsMarkov} at arbitrary temperatures, in the following sense: any noise $\CN_A$ on local region $A$ can be approximately undone by a quasi-local recovery map $\CR$,
\begin{align}
\vrho \approx \CR\circ \CN_A[\vrho] \quad \text{for the Gibbs state}\quad \vrho = \frac{e^{-\beta \vH}}{\tr[e^{-\beta \vH}]}.
\end{align}
This gives an operational sense in which the quantum correlations between region $A$ and the rest of the system are mediated by a quasi-local neighborhood of $A$, or equivalently by the decay of \textit{quantum conditional mutual information}~\cite{fawzi2015quantum}. A key conceptual shift from the classical picture is that the underlying Gibbs sampling dynamics~\cite{chen2025efficient} provides a more natural, and perhaps essential, viewpoint on quantum conditional independence without explicitly pinning or measuring a subsystem. This dynamical perspective also motivates extending the Markov property to a broader class of states, namely metastable states, i.e., approximate stationary states $\CL[\vsigma]\approx 0$ of the detailed-balanced Lindbladian dynamics~\cite{bergamaschi2025structural}:
\begin{align}
\text{approximate stationarity}\quad \quad \CL[\vsigma]&\approx 0\quad \\
\text{implies approximate Markov property}\quad\quad \vsigma &\approx \CR\circ \CN_A[\vsigma].
\end{align}
A natural question raised in~\cite{bergamaschi2025structural} is how this noise--recovery property should be interpreted operationally, as a form of robustness or \textit{reusability} of naturally occurring quantum states.\footnote{A revised version of~\cite{bergamaschi2025structural} is in preparation.} In fact, a Markov state can be surprisingly \textit{fragile} to local observers who wish to learn about it. Suppose we perform a local measurement and apply the recovery map. Then the Markov property only guarantees recovery when the measurement outcomes are \textit{discarded}. That is, if we record the measurement outcomes—thereby gaining information about the state—the post-measurement state need not return to the original state.

To address this fragility, \cite{bergamaschi2025structural} proposed a strengthened recovery condition requiring recovery to hold for each post-selected state associated with individual measurement outcomes (i.e., for each Kraus operator). This property, known as the \textit{strong Markov property}, takes the form
\begin{align}
\CR[ \vK \vsigma \vK^{\dagger}] \approx \vsigma \cdot \tr[\vK\vsigma \vK^{\dagger}]
\quad \text{for any local operator $\vK$}.
\end{align}
Operationally, this captures the possibility of repeatedly performing local measurements and recovering the \textit{same} underlying state. Physically, since the recovery map can be implemented by a Lindbladian system-bath evolution~\cite{SA24}, it provides the conceptual intuition for quantum states that remain robust and \textit{stable} in open thermal environments.

While the above proposal in \cite{bergamaschi2025structural} left the mathematical characterization of such states open, this note establishes the formal criteria for the strong Markov property:

\begin{thm}[Informal]
A metastable state satisfies correlation decay if and only if there exists an approximate recovery map under local post-selected measurements.
\end{thm}

We further establish several structural and operational consequences of strongly Markov states. First, this property enables the estimation of multiple observables from a \textit{single copy} of the state via a repeated measurement--recovery protocol. %Since the recovery map is itself a Lindbladian evolution, this also explains how a system may relax back to equilibrium through interaction with a thermal bath.

Second, pairs of strongly Markov states under the same Lindbladian must be either locally very close or locally very far apart. This phenomenon follows from the single-copy tomography property. Indeed, if two strongly Markov states have slightly different local marginals, the repeated measurement--recovery protocol \textit{amplifies} this difference and produces a \textit{single-shot} distinguisher between them. Because the entire measurement--recovery procedure is quasi-local, the two states must already have been distinguishable within a quasi-local neighborhood. Consequently, the set of local marginals of strongly Markov metastable states is constrained to form well-separated clusters.

Third, strongly Markov states under the same Lindbladian exhibit a form of local extremality: a mixture of two such strongly Markov states with locally distinguishable marginals cannot itself be strongly Markov. In contrast, metastability and the local Markov property are preserved under convex combinations. Put differently, if a strongly Markov state can be expressed as a mixture of two strongly Markov states, then their local marginals must be nearly indistinguishable. This behavior is reminiscent of quantum error-correcting codes, where logical quantum information is accessible only through global measurements and remains invisible locally. Our derivation suggests that such a structure is not unique to highly structured quantum memories. Rather, it appears generically in systems with noncommuting Hamiltonians and in metastable states that need not store quantum information.

\subsection{Main formulations}

Our main results concern the generic relation between several dynamical and static properties of quantum states.

\begin{figure}[t]
\centering
\begin{tikzpicture}[
box/.style={
    draw,
    rounded corners,
    align=center,
    minimum width=3.2cm,
    minimum height=0.9cm
},
arrow/.style={-{Triangle[length=2mm,width=2mm]}, thick},
line/.style={thick}
]

\node[box] (meta)    at (0,2)  {Approximate\\Stationarity};
\node[box] (cluster) at (0,-2) {Clustering};
\node[box] (markov)  at (5,2)  {Local Markov\\Property};
\node[box] (strong)  at (5,0)  {Strong Local\\Markov Property};

\node[box] (ext) at (10,1.6)  {Local\\Extremality};
\node[box] (tom) at (10,0)    {Single-Copy\\Tomography};
\node[box] (sep) at (10,-1.8) {Locally\\Separated};

\draw[arrow] (meta.east) -- (markov.west);

\coordinate (merge) at (2.7,0);
\draw[line]  (meta.east) -- (merge);
\draw[line]  (cluster.east) -- (merge);
\draw[arrow] (merge) -- (strong.west);

\draw[arrow] (strong.south) |- (cluster.east);

\draw[arrow] (strong.north east) -- (ext.west);

\draw[arrow] (strong.east) -- (tom.west);

\draw[arrow] (tom.south) -- (sep.north);

\end{tikzpicture}

\caption{Relations between dynamical and static properties of a quantum state: approximate stationarity, clustering, and (strong) local Markov properties, together with consequences of the strong local Markov property.}
\end{figure}
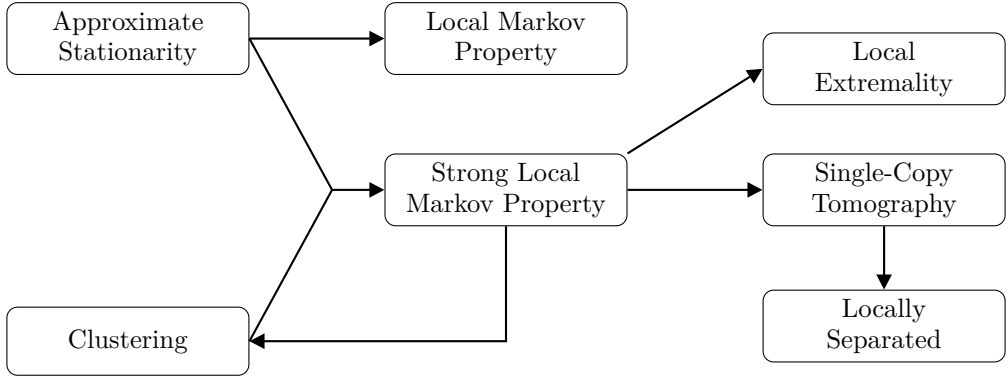

\begin{itemize}
    \item \textbf{Metastability:} Approximate stationarity under coupling to a finite-temperature bath.
    \item \textbf{Clustering:} Certain pairs of far-apart observables are approximately uncorrelated.
    \item \textbf{(Local) Markov property:} Local noise can be approximately recovered quasi-locally.
    \item \textbf{Strong local Markov property:} Local measurement outcomes can be recovered quasi-locally, conditioned on each outcome.
\end{itemize}

Our results show that the above notions admit natural formulations under which their relations become sharp. We begin by introducing the precise definitions used throughout this note.

The starting point comes from the study of metastable states~\cite{bergamaschi2025structural}, which characterizes approximate stationary states of Lindbladian dynamics (see~\autoref{section:the_lindbladian})
\begin{align}
\norm{\frac{\rd \vsigma}{\rd t}}_1 \approx 0 \quad \text{where}\quad
    \frac{\rd \vsigma}{\rd t}  = \CL[\vsigma] = -i[\vH,\vsigma] +\sum_a \CL_a[\vsigma].
\end{align}
The following alternative formulation, which can be converted back and forth with $\norm{\CL[\vsigma]}_1 \approx 0$ or $\norm{\CL_a[\vsigma]}_1, \norm{[\vH,\vsigma]}_1\approx 0$~\cite{bergamaschi2025structural}, will be more directly applicable for our purposes.

\begin{defn}
    [Approximate Detailed Balance Condition]
    \label{defn:intro_adb}
    We say a state $\vsigma$ satisfies approximate detailed balance under $\vA^a$ with error
    \begin{equation}
        \mathsf{ADB}_a[\vsigma]
        :=
        \iint_{-\infty}^{\infty}
        \lnormp{\vA^a(\omega,t)\sqrt{\vsigma}
        -
        \sqrt{\vsigma}\vrho^{-\frac{1}{2}}\vA^a(\omega,t)\vrho^{\frac{1}{2}}}{2}^2
        \gamma(\omega) g(t)\,
        \rd \omega \rd t,
    \end{equation}
    where $\gamma(\omega)$ is the shifted Metropolis weight~\eqref{eq:Metropolis}, and $g(t) := \frac{1}{\beta \cosh(2\pi t/\beta)}$.
\end{defn}

The main result of~\cite{bergamaschi2025structural} shows that approximate detailed balance implies a local Markov property. To obtain the stronger recovery property, the additional ingredient we identify is a correlation decay condition, commonly referred to as \textit{clustering}. In this note, we denote by $\Lambda$ the full set of qubits and consider tripartitions $ABC = \Lambda.$

\begin{defn}[Clustering]
    We say a state $\vsigma$ satisfies $\epsilon$-clustering for regions $A, C \subset \Lambda$ if for all operators supported on those regions,
    \begin{align}
        \labs{
        \tr[\vsigma\vX_A\vY_C]
        -
        \tr[\vsigma\vX_A]\tr[\vsigma\vY_C]
        }
        \le
        \epsilon
        \norm{\vX_A}\norm{\vY_C},
    \end{align}
    for every pair of operators $\vX_A$, $\vY_C$ supported on $A$ and $C$, respectively.
\end{defn}

There are also several alternative formulations of correlation decay for quantum states that are not a priori equivalent.\footnote{We thank Sengqi Sang and Daniel Ranard for helpful discussions.}
\begin{rmk}
    An approximate stationary state may cluster for small $A$ and large $C$, but not for large $A,C$, such as for logical states in the 4D toric code below a constant temperature.
\end{rmk}

Next, we introduce the Markov properties.
\begin{defn}[Local Markov property]
    A state $\vsigma$ satisfies an $\epsilon$-local Markov property for region $A$ if there exists a channel $\CM_{AB}$ supported on $AB$ such that for any channel $\CN_A$ supported on $A$,
\begin{align}
    \lnorm{
    \CM_{AB}\circ \CN_{A}[\vsigma]
    -
    \vsigma
    }_1
    \le
    \epsilon.
\end{align}
\end{defn}
In the above, the map $\CN_A$ is completely positive and trace-preserving (CPTP). In the case of the strong Markov property, this is further relaxed.
\begin{defn}[Strong local Markov property]
\label{defn:strongMarkov}
A state $\vsigma$ satisfies an $\epsilon$-strong local Markov property for region $A$ if there exists a channel $\CM_{AB}$ supported on $AB$ such that for any operator $\vK$ supported on $A$ with $\norm{\vK}\le1$,
\begin{align}
    \lnorm{
    \CM_{AB}[ \vK \vsigma \vK^{\dagger}]
    -
    \vsigma \cdot \tr[\vK\vsigma \vK^{\dagger}]
    }_1
    \le
    \epsilon.
\end{align}
\end{defn}

Since the operator $\vK\vsigma\vK^\dagger$ need not have unit trace, we need the normalization factor $\tr[\vK\vsigma\vK^\dagger]$. While this formulation is convenient for analysis, it is not operationally the most transparent; we therefore introduce a nearly equivalent form in which the Kraus operators arise from a measurement.

\begin{defn}[Strong local Markov property for measurements]
\label{defn:strongMarkov_measurement}
We say a state $\vsigma$ satisfies an $\epsilon$-strong local Markov property for measurements if there exists a channel $\CM_{AB}$ supported on $AB$ such that for any collection of Kraus operators supported on $A$ satisfying $\sum_i \vK_i^\dagger \vK_i = \vI$,
\begin{align}
    \sum_i
    \lnorm{
    \CM_{AB}[ \vK_i \vsigma \vK_i^{\dagger}]
    -
    \vsigma \cdot \tr[\vK_i\vsigma \vK_i^{\dagger}]
    }_1
    \le
    \epsilon .
\end{align}
\end{defn}

Operationally, this condition means that after performing any local measurement, the state can be approximately recovered regardless of which outcome $i$ is observed. One can pass from~\autoref{defn:strongMarkov} to~\autoref{defn:strongMarkov_measurement} at the cost of a multiplicative factor proportional to the number of Kraus operators\footnote{A finite-dimensional channel admits a finite Kraus representation. However, when describing outcome statistics of general measurements there is no \emph{a priori} bound on the number of Kraus operators.}. Moreover, by the triangle inequality of the 1-norm, strong local Markov for measurements always implies the original Markov property.

In our context, the recovery map is naturally taken to be the detailed-balance Lindbladian dynamics~\cite{chen2023efficient}.

\begin{figure}[t]
\centering
\begin{tikzpicture}[x=1.25cm,y=1.1cm,>=latex]

\node[anchor=west] at (-4.8,0) {(1) Markov property};

\node at (0,0) {$\vrho$};
\draw[->] (0.35,0) -- (1.0,0);

\node[draw,minimum width=1.0cm,minimum height=0.7cm] (CN) at (1.8,0) {$\CN$};

\draw[dashed,thick,->] (1.8,-0.35) -- (2.2,-0.75);
\node[right] at (2.2,-0.75) {$i$};

\draw[->] (2.3,0) -- (2.9,0);

\node[draw,minimum width=1.0cm,minimum height=0.7cm] (R1) at (3.7,0) {$\CR$};

\draw[->] (4.2,0) -- (4.9,0);
\node[right] at (4.9,0) {$\vrho$};

\node[anchor=west] at (-4.8,-1.8) {(2) Strong Markov property};

\node at (0,-1.8) {$\vrho$};
\draw[->] (0.35,-1.8) -- (1.0,-1.8);

\node[draw,minimum width=1.0cm,minimum height=0.7cm] (CE) at (1.8,-1.8) {$\CE$};

\draw[thick,->] (1.8,-2.15) -- (2.2,-2.55);
\node[right] at (2.2,-2.55) {$i$};

\draw[->] (2.3,-1.8) -- (3.0,-1.8);
\node at (3.8,-1.8) {$\vK_i \vrho \vK_i^\dagger$};
\draw[->] (4.6,-1.8) -- (5.2,-1.8);

\node[draw,minimum width=1.0cm,minimum height=0.7cm] (R2) at (6.0,-1.8) {$\CR$};

\draw[->] (6.5,-1.8) -- (7.2,-1.8);
\node[right] at (7.2,-1.8) {$\propto \vrho$};

\end{tikzpicture}

\caption{(1) The Markov property: a recovery map $\CR$ approximately recovers the state from the averaged noise channel $\CN$. (2) The strong Markov property: each post-selected branch can be approximately recovered, up to normalization.}
\label{fig:strong-markov-gadget}
\end{figure}
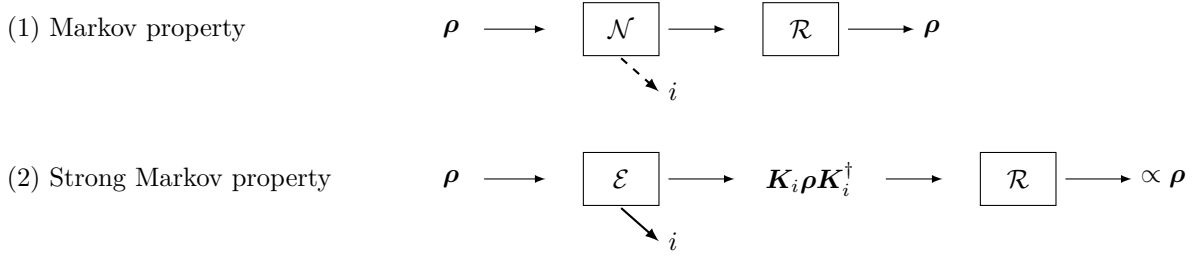

\subsection{Main results}

We now state the main structural relations between the above notions.
\begin{lem}[Clustering + metastability implies strong Markov]
\phantomsection
\label{lem:main}
Consider a Hamiltonian $\vH$ on $n$ qubits with interaction degree $d$ and inverse temperature $\beta>0$. Consider a tripartition $ABC=\Lambda$. Suppose a state $\vsigma$ satisfies
\begin{enumerate}
\item $\epsilon_{AC}$-clustering between regions $A$ and $C$, and
\item $\epsilon_{ADB}$-approximate detailed balance for all single Pauli jumps on region $AB$
\end{enumerate}

Then there is a family of approximate quasi-local recovery maps $\{\CR_{AB,t}\}_{t \ge 0}$ such that
    \begin{align}
        \lnorm{\CR_{AB,t}[\vK_A\vsigma \vK_A^{\dagger}]-\tr[\vK_A\vsigma \vK_A^{\dagger}]\cdot \vsigma}_1 \le \epsilon_{AC}+ \e^{\mu|AB|}\,\cdot \big(t^{-1}+\sqrt{\epsilon_{\mathsf{ADB}}}\big)^{\lambda}+ c\labs{AB} t \sqrt{\epsilon_{ADB}}  \quad \text{for each}\quad \vK_A, \norm{\vK_A}\le 1,
    \end{align}
    for some $\mu >0$ and $0<\lambda<1$ depending only on $\beta,d$ and an absolute constant $c$.
\end{lem}
Therefore, for sufficiently large time $t$ and small $\epsilon_{ADB}$, the error is limited by that of clustering. Specifically, the recovery map is entirely defined by the Hamiltonian and is independent of the metastable states. A caveat of the recovery guarantee, however, is that the recovery time depends on the desired error, which may be quasi-polynomial. This scaling is rooted in the worst-case analysis of a certain mixing guarantee~\cite{bergamaschi2025structural} on the quasi-local neighborhood, which we believe could be much better depending on the specific physics and phase of the system.

The above implication admits a complementary converse, relating recovery to correlation decay.

\begin{lem}[Strong Markov implies clustering]
Consider a tripartition $ABC = \Lambda$. Suppose that a state $\vsigma$ satisfies the $\epsilon$-strong local Markov property for operators supported on region $A$, with an approximate recovery map $\CM$ acting on $AB$. Then the state satisfies $4\epsilon$-clustering between $A$ and $C$.
\end{lem}

Here, the recovery map need not satisfy detailed balance. However, if strong Markov holds with \textit{any} recovery map, we must have correlation decay, and hence, the natural Lindbladian dynamics must also work, with different parameters.

In the above, the relation between clustering and strong Markov is stated pointwise for each specific annulus tripartition. Therefore, the quantum state may still exhibit global, long-range correlation but nevertheless satisfies strong Markov properties for local neighborhoods--- such as a logical state of a quantum code.

The strong Markov property has several operational and structural consequences. Indeed, when each measurement outcome can be individually recovered, we can obtain the marginals of the state by repeatedly measuring and recovering from a single copy of the state, as the consecutive outcomes are effectively i.i.d.
\begin{lem}[Strongly Markov states are repeatable]
\label{lem:repeatable_main}
Suppose a state $\vsigma$ satisfies the $\epsilon$-strong local Markov property for measurement channel $\CK = \sum_{i=1}^k \vK_i[\cdot]\vK_i^{\dagger}$ supported on region $A\subset \Lambda$ with recovery channel $\CM$.
Then for each outcome $i$, the empirical mean estimator $\hat{\mu}_i$ from repeated measurement--recovery satisfies
\begin{align}
         \Pr\L(\labs{\hat{\mu}_i - \tr[\vK_i^{\dagger}\vK_i\vsigma]}\ge \tau\R) \le 2\exp\L(-2r\tau^2\R) + r \epsilon
\end{align}
while returning a state $\vsigma'$ such that $\norm{\vsigma-\vsigma'}_1 \le r \epsilon.$
\end{lem}
The higher the quality of the strong Markov property, the more rounds we can repeatedly recover. This also immediately allows us to distinguish two strongly Markov states with slightly different marginals by repeated measurement to sufficient precision. Consequently, the marginals must differ substantially in the neighborhood on which the recovery map is approximately supported.

\begin{lem}[Locally close but distinguishable]
\label{lem:local_close}
Consider any two $\epsilon$-strongly Markov states $\vsigma_1$ and $\vsigma_2$ under the same recovery map $\CM$. Suppose that the marginal on $A$ is at least $\delta$ far in trace distance $\norm{\tr_{BC}[\vsigma_1]-\tr_{BC}[\vsigma_2]}_1 \ge \delta$. Then there is a measurement-recovery procedure that distinguishes the case $\vsigma_1$ from $\vsigma_2$ with failure probability $\frac{16\epsilon}{\delta^2}\log ( \frac{e\delta^2}{8\epsilon})+2\epsilon.$
\end{lem}

Moreover, strongly Markov states behave like extremal points with respect to local marginals.

\begin{lem}[Locally extremal]
Consider a region $A$ and a recovery map $\CM$ for which an $\epsilon$-strongly Markov state $\vsigma$ is a mixture over two $\epsilon$-strongly Markov states
\begin{align}
    \vsigma = p_1\vsigma_1 + p_2 \vsigma_2.
\end{align}
Then the local marginals are close
\begin{align}
\normp{\vsigma^{(A)}_1-\vsigma_2^{(A)}}{1} \le 2\sqrt{\frac{2\epsilon}{p_1p_2}}.
\end{align}
\end{lem}

 \acknowledgments
We thank Thiago Bergamaschi and Umesh Vazirani, András Gilyén, Robbie King, and Cambyse Rouzé for inspiring collaborations~\cite{bergamaschi2025structural, chen2025catalytic,chen2026efficient, chen2025GibbsMarkov}. We thank David Gamarnik, Sarang Gopalakrishnan, Aram Harrow, Vedika Khemani, Isaac Kim, Alexei Kitaev, Tomotaka Kuwahara, Lin Lin, Tony Metger, Daniel Ranard, Sengqi Sang, and Alexander Zlokapa for helpful discussions. Near the completion of this note, we became aware of independent related work by Zhi Li, Raz Firanko, and Timothy H.~Hsieh.

\section{Preliminary}
\label{section:the_lindbladian}
Central to our study is the KMS-detailed-balance Lindbladian family of \cite{chen2023efficient,chen2023quantum}; its explicit form will be used throughout the paper. Fix a Hamiltonian $\vH$ on $n$ qubits with interaction degree $d$, an inverse temperature $\beta>0$, and a single self-adjoint jump $\vA^a = \vA^{a\dagger}$. We consider the (quasi-local) Lindbladian defined by
\begin{align}
		\CL_a[\cdot] = \underset{\text{``coherent''}}{\underbrace{ -\ri [\vC^a, \cdot]}} +
		\int_{-\infty}^{\infty} \gamma(\omega) \bigg(\underset{\text{``transition''}}{\underbrace{\hat{\vA}^a(\omega)(\cdot)\hat{\vA}^{a}(\omega)^\dagg}} - \underset{\text{``decay''}}{\underbrace{\frac{1}{2}\{\hat{\vA}^{a}(\omega)^\dagg\hat{\vA}^a(\omega),\cdot\}}}\bigg)\rd\omega\label{eq:exact_DB_L}
\end{align}
\noindent with the shifted-Metropolis weight
\begin{align}
    \gamma(\omega) = \exp\L(-\beta\max\left(\omega +\frac{\beta \sigma^2}{2},0\right)\R).\label{eq:Metropolis}
\end{align}
\noindent The central ingredient in the Lindbladian above is the operator Fourier transform \cite{chen2023quantum,chen2023efficient}. The operator FT of an operator $\vA$, associated to the Hamiltonian $\vH$, can be written as
\begin{align}\label{eq:OFT}
{\hat{\vA}}_\sigma(\omega)=  \frac{1}{\sqrt{2\pi}}\int_{-\infty}^{\infty} \e^{\ri \vH t} \vA \e^{-\ri \vH t} \e^{-\ri \omega t} f_{\sigma}(t)\rd t\quad \text{with}\quad f_{\sigma}(t) := e^{-\sigma^2t^2}\sqrt{\sigma\sqrt{2/\pi}}
    \end{align}

\noindent where the function $f_\sigma(t)$ is a Gaussian filter of energy \textit{width} $\sigma$, which we take to be $\sigma =1/\beta$ throughout this note. The ``coherent part'' $\vC^a$ is a Hermitian operator
\begin{align}
    &\vC^a := \lim_{\eta \rightarrow 0}\int_{-\infty}^{\infty}\int_{\labs{t} \ge \eta} \gamma(\omega)  c(t) \cdot \hat{\vA}^a(\omega,t)^{\dagger}\hat{\vA}^a(\omega,t)  \rd t\rd \omega, \quad  \\\text{with}\quad  &c(t) := \frac{1}{\beta\sinh(2\pi t/\beta)} \quad \text{and}\quad  \hat{\vA}^a(\omega,t):=\e^{i\vH t}\hat{\vA}^a(\omega)\e^{-i\vH t}.
\end{align}

The main connection between the Lindbladian and the Markov properties is that the time-averaged Lindbladian dynamics gives a natural recovery map
\begin{align}
    \CR_{A,t}[\cdot]:= \frac{1}{t}\int_{0}^t \exp\L(s\,\sum_{a\in P^1_A}\CL_a\R)[\cdot] \,\rd s\label{eq:RAt_intro}
\end{align}
for generators $\CL_a$ with jumps $\vA^a \in P^1_A$ ranging over all single-site Pauli operators ($\vX_i,\vY_i,\vZ_i$) acting on the region $A$.

\begin{thm}[Approximate Detailed Balance implies a Local Markov Property{\cite{bergamaschi2025structural}}]\label{thm:local_recovery}
Consider a Hamiltonian $\vH$ on $n$ qubits with interaction degree $d$ and inverse temperature $\beta>0$. Suppose a state $\vsigma$ satisfies approximate detailed balance (\autoref{defn:intro_adb}) for all single-site Pauli operators on $\mathsf{A}$, with error $\epsilon_{\mathsf{ADB}} := \max_{\vP\in P^1_\mathsf{A}} \mathsf{ADB}_{\vP}[\vsigma].$ Then for any $t>0$, the time-averaged dynamics $\CR_{\mathsf{A},t}$ \eqref{eq:RAt_intro} defines an approximate recovery map for $\vsigma$, with error
 \begin{align}
     \norm{\vsigma-\CR_{\mathsf{A},t}[\CN_{\mathsf{A}}[\vsigma] ]}_1 \le \e^{\mu|\mathsf{A}|}\,\cdot t^{-\lambda} + c|\mathsf{A}|\cdot t\cdot \epsilon_{\mathsf{ADB}}^{1/2}
 \end{align}
 \noindent for some $0<\mu <\mathsf{poly}(\beta, \beta^{-1}, d) $ and $1 > \lambda> 1/\mathsf{poly}(\beta, \beta^{-1}, d)$ and an absolute constant $c$.
\end{thm}

\section{Proof of strong Markov property and converse}

We now present the proof of the strong Markov property. We state the assumptions explicitly and restate the lemma.

\begin{lem}[Clustering + approximate stationarity implies strong Markov] Consider a Hamiltonian $\vH$ on $n$ qubits with interaction degree $d$ and inverse temperature $\beta>0$. Consider a tripartition $ABC= \Lambda$. Suppose that a state $\vsigma$ satisfies (1) clustering for all operators supported on regions $A$ and $C$
    \begin{align}
        \labs{ \tr[\vsigma\vX_{A}\vY_{C}] -\tr[\vsigma\vX_{A}]\tr[\vsigma\vY_{C}] } \le \epsilon_{AC}\norm{\vX_{A}}\norm{\vY_{C}} \quad \text{for each}\quad \vX_A, \vY_C \quad \text{supported on}\quad A,C
    \end{align}
     and (2) approximate detailed balance for all single Pauli jumps on a region $AB$
\begin{align}
    \max_{a\in P^1_{AB}} \mathsf{ADB}_a[\vsigma] \le \epsilon_{ADB}.
\end{align}
    Then there is a family of approximate quasi-local recovery maps $\{\CR_{AB,t}\}_{t \ge 0}$ such that
    \begin{align}
        \lnorm{\CR_{AB,t}[\vK_A\vsigma \vK_A^{\dagger}]-\tr[\vK_A\vsigma \vK_A^{\dagger}]\cdot \vsigma}_1 \le \epsilon_{AC}+ \e^{\mu|AB|}\,\cdot \big(t^{-1}+\sqrt{\epsilon_{\mathsf{ADB}}}\big)^{\lambda}+ c\labs{AB} t \sqrt{\epsilon_{ADB}}  \quad \text{for each}\quad \vK_A, \norm{\vK_A}\le 1,
    \end{align}
    for some $\mu >0$ and $0<\lambda<1$ depending only on $\beta,d$ and an absolute constant $c$.
\end{lem}
    Approximate detailed balance is most natural for directly dealing with $\vsigma$-weighted norms; notice that this is needed for all single Pauli operators supported on the region $AB$ (not just on $A$).

\begin{proof}
    Consider the time-averaged dynamics $\CR_{AB,t}$ associated with the Lindbladian constructed from jumps acting on $AB$. To show the main claim in trace distance, we consider any test operator $\vX$ with $\norm{\vX} =1$ and use the duality
\begin{align}
    \left\| \CR_{AB,t}[\vK \vsigma \vK^\dagger] - \vsigma \, \tr[\vK \vsigma \vK^\dagger] \right\|_1
=
\sup_{\|\vX\| \le 1}
\left|
\tr\!\left[ \vK \vsigma \vK^\dagger \, \CR_{AB,t}^\dagger[\vX] \right]
-
\tr[\vK \vsigma \vK^\dagger] \, \tr[\vsigma \vX]
\right|.
\end{align}

The main calculation is the following chain of approximations, up to (Claim)
\begin{align}
    \tr[\vK\vsigma \vK^{\dagger} \CR^{\dagger}_{AB,t}[\vX]  ]
    &\approx \tr[ \vK\vsigma \vK^{\dagger} (\vI\otimes \vX'_{C})] \tag*{(Claim)}\\
    &= \tr[ \vsigma \vK^{\dagger}\vK \vX'_{C}] \tag*{(Since $[\vK,\vX'_C]=0$)} \\
    &\approx \tr[\vsigma\vK^{\dagger}\vK] \tr[\vsigma \vX'_{C}] \tag*{(Assumption: clustering of correlation)}\\
    &\approx \tr[\vsigma\vK^{\dagger}\vK] \tr[\vsigma\CR_{AB,t}^{\dagger}[\vX]] \tag*{(Claim once again, with $\vK' = \vI$)}\\
    &\approx \tr[\vsigma\vK^{\dagger}\vK] \tr[\vsigma \vX]
\end{align}

where the last line uses~\cite[Theorem C.2]{bergamaschi2025structural} so that
\begin{align}
    \norm{\CR_{AB,t}[\vsigma]-\vsigma} \le t \sum_{a\in P^1_{AB}}\norm{\CL_a[\vsigma]}_1\le c t\labs{AB} \sqrt{\max_{a\in P^1_{AB}} \mathsf{ADB}_a[\vsigma]}  = c\labs{AB} t \sqrt{\epsilon_{ADB}}
\end{align}
for some absolute constant $c$.

It remains to show the claim that the recovery map trivializes the operator on $AB$
    \begin{align}
        \CR^{\dagger}_{AB,t}[\vX] \approx \vI_{AB}\otimes \vX_C'
    \end{align}
    in a suitable norm, and we follow an argument close to~\cite[Lemma D.4]{bergamaschi2025structural}.
It suffices to show
\begin{align}
\CR^{\dagger}_{AB,t}[\vX] \approx \undersetbrace{=:\vI_{AB}\otimes \vX_C'}{\CN_{AB}^{\dagger}\circ\CR^{\dagger}_{AB,t}[\vX]}
\end{align}
  where we define the forgetful channel
\begin{align}
    \CN_{AB}[\cdot] = \frac{\vI_{AB}}{\tr[\vI_{AB}]} \otimes \tr_{AB}[\cdot] = \frac{1}{\labs{P_{AB}}}\sum_{i\in P_{AB}} \vV_i[\cdot]\vV^{\dagger}_i.
\end{align}

Consider the difference evaluated on $\vK\vsigma\vK^{\dagger}$

\begin{align}
    {\tr\big[ \vK\vsigma \vK^{\dagger} \undersetbrace{=:\vI_{AB}\otimes \vX_C'}{\CN_{AB}^{\dagger}\circ\CR^{\dagger}_{AB,t}[\vX]} \big] - \tr[\vK\vsigma \vK^{\dagger} \CR^{\dagger}_{AB,t}[\vX]  ]}
&= \frac{1}{\labs{P_{AB}}} {\sum_i \tr\L[ \vK\vsigma \vK^{\dagger} \L( [\vV_i,\CR^{\dagger}_{AB,t}[\vX]]\vV_i^{\dagger} +  \vV_i[\CR^{\dagger}_{AB,t}[\vX],\vV_i^{\dagger}]\R)  \R]}
\end{align}
using the commutator expression
\begin{align}
\CN_{AB}^{\dagger}[\vO]-\vO =   \frac{1}{\labs{P_{AB}}} \sum_{i\in P_{AB}} [\vV_i, \vO]\vV_i^{\dagger}+\vV_i[\vO,\vV_i^{\dagger}].
\end{align}
We further brute-force rewrite the operator $\vK$ into local Pauli strings, using the following elementary expansion
\begin{align}
    \vK &= \sum_{a\in P_A} c_a\vP_a\quad \text{such that}\quad \labs{c_a} \le \norm{\vK} \le 1
\end{align}
and rewrite the $\vV_i$ as a product of individual local Paulis
\begin{align}
    \vV_i &= \vA^1\cdots \vA^{\ell} \\
    & = \undersetbrace{=:\vB^j}{\vA^1\cdots \vA^{j-1}} \vA^j \undersetbrace{=:\vC^j}{\vA^{j+1}\cdots \vA^{\ell}} \quad \text{for any}\quad j.
\end{align}

Therefore,
\begin{align}
    \tr\L[ \vK\vsigma \vK^{\dagger} [\vV_i,\CR^{\dagger}_{AB,t}[\vX]]\vV_i^{\dagger} \R] &= \sum_j \tr\L[ \vK\vsigma \vK^{\dagger} \vB^j[\vA^j,\CR^{\dagger}_{AB,t}[\vX]]\vC^j\vV_i^{\dagger} \R] \\
    &=\sum_{a,b\in P_A}\sum_j c_ac_b^*\tr\L[ \vP_a\vsigma \vP_b \vB^j[\vA^j,\CR^{\dagger}_{AB,t}[\vX]]\vC^j\vV_i^{\dagger} \R].
\end{align}
As in \cite[Proof of Theorem IV.6]{bergamaschi2025structural}, we have that for $c_2,c_1,\alpha_1$ as function of $\beta,d$,
\begin{align}
    \max_{\|\vX\|\leq 1}\max_{\vP,\vQ\in P_\mathsf{AB}\cup \vI} \max_{\vA\in P_\mathsf{AB}^1}  \labs{\tr\L[ \vsigma \vP \frac{1}{2}[\vA^a,\CR^{\dagger}_{\mathsf{AB},t}[\vX]]\vQ\R]}&\le (2c_2\delta_t^{c_1})^{\alpha_1} \alpha_2^{|\mathsf{\mathsf{AB}}|} \quad \text{for}\quad \delta_t := \frac{1}{t} + \max_{\vA\in P^1_\mathsf{A
    B}} \mathsf{ADB}_{\vA}[\vsigma]^{1/2}.
\end{align}
Summing over all possible $a,b \in P_A$ and $j$ contributes $\sum_{a\in P_A} \labs{c_a} \le 4^{\labs{A}}$; therefore, there exists an updated $\alpha'_2$ such that
\begin{align}
    \labs{\sum_{a,b\in P_A}\sum_j c_ac_b^*\tr\L[ \vP_a\vsigma \vP_b \vB^j[\vA^j,\CR^{\dagger}_{AB,t}[\vX]]\vC^j\vV_i^{\dagger} \R]} \le \labs{AB} \cdot 16^{\labs{A}} (2c_2\delta_t^{c_1})^{\alpha_1} \alpha_2^{|\mathsf{\mathsf{AB}}|} \le (2c_2\delta_t^{c_1})^{\alpha_1} (\alpha'_2)^{|\mathsf{\mathsf{AB}}|},
\end{align}
proving the advertised claim.
Collect the error terms to conclude the proof.

\end{proof}
In the above, $\CR_{AB,t}$ is stated as a quasi-local recovery map that may extend outside $AB$, which can be truncated by standard Lieb-Robinson bounds.

Next, we show the converse statement, and in fact, the recovery channel need not be detailed-balance.

\begin{lem}[Strong Markov implies clustering]
    Suppose that a state $\vsigma$ satisfies the strong local Markov property for operators on a region $A$ and an approximate recovery map $\CM$ on $AB$ such that
\begin{align}
    \lnorm{\CM[ \vK_A \vsigma \vK_A^{\dagger}] -\vsigma \cdot \tr[\vK_A\vsigma \vK_A^{\dagger}]}_1  \le \epsilon_{AB} \quad \text{for each}\quad \vK_A\quad \text{such that}\quad \norm{\vK_A}\le 1.
\end{align}
    Then for any operator $\vX_{A}$, $\norm{\vX_{A}}\le 1$ acting on $A$ and operator $\vY_{C}$, $\norm{\vY_{C}}\le 1$ acting on the complement $C= (AB)^c,$ we must have clustering
    \begin{align}
        \labs{\tr[\vsigma \vX_{A}\vY_{C}] - \tr[\vsigma \vX_{A}]\tr[\vsigma\vY_{C}]} \le 4 \epsilon_{AB}.
    \end{align}
\end{lem}
\begin{proof}

For any bounded operator $\vX_{A}$, $\norm{\vX_{A}}\le1$, consider the decomposition $\vX_{A} = \vK_1^{\dagger}\vK_1 - \vK_2^{\dagger}\vK_2 + i\vK_3^{\dagger}\vK_3 - i\vK_4^{\dagger}\vK_4 $ such that $\norm{\vK_i^{\dagger}\vK_i}\le \norm{\vX_{A}} =1$. Then for each $\vK_i,$
\begin{align}
    \tr[\vY_{C} \vK_i^{\dagger}\vK_i\vsigma ] = \tr[\vY_C \vK_i\vsigma \vK_i^{\dagger} ] &=  \tr[\CM_{AB} [\vY_C \vK_i\vsigma \vK_i^{\dagger} ]]\\
    & =\tr[\vY_C\CM_{AB} [\vK_i\vsigma \vK_i^{\dagger}]]\\
    & \stackrel{\epsilon_{AB}}{\approx}\tr[\vY_C \vsigma ] \cdot \tr[\vK_i\vsigma \vK_i^{\dagger}]= \tr[\vY_C \vsigma ] \cdot \tr[\vK_i^{\dagger}\vK_i\vsigma ],
\end{align}
which concludes the proof.
\end{proof}

\section{Implications of Strong Markov property}

Now that we have derived the strong Markov property from natural assumptions, we derive several consequences. Algorithmically, this implies a single-copy tomography with a provable guarantee depending on the quality of the strong Markov property. Structurally, the strong Markov property (together with the fact that the recovery is independent of the states) also implies constraints on the geometry of the set of viable local marginals.

\subsection{Single-copy tomography}
The strong Markov property is intimately related to our ability to repeatedly measure and restore the state multiple times. We will introduce a slight variant of the strong Markov property that is more natural when thinking about measurement protocols.

\begin{lem}[Strongly Markov states can  be repeatedly measured]\label{lem:repeatable}
Suppose a state $\vsigma$ satisfies the $\epsilon$-strong local Markov property for measurement channel $\CK = \sum_{i=1}^k \vK_i[\cdot]\vK_i^{\dagger}$ supported on region $A\subset \Lambda$ with recovery channel $\CN$.
Then for each outcome $i$, the empirical mean estimator $\hat{\mu}_i$ from $r$ rounds of repeated measurement--recovery satisfies
\begin{align}
         \Pr\L(\labs{\hat{\mu}_i - \tr[\vK_i^{\dagger}\vK_i\vsigma]}\ge \tau\R) \le 2\exp\L(-2r\tau^2\R) + r \epsilon
\end{align}
while returning a state $\vsigma'$ such that trace distance $\norm{\vsigma-\vsigma'} \le r \epsilon.$

Therefore, we can estimate $\tr[\vK_i^{\dagger}\vK_i\vsigma]$ to error $\tau$ using $r = \lceil\frac{\log(4\tau^2/\epsilon)}{2\tau^2}\rceil$ rounds of measurement--recovery at a failure probability $\frac{\epsilon}{2\tau^2}\log ( \frac{4e\tau^2}{\epsilon})+\epsilon $.
\end{lem}

\subsubsection{Proof of~\autoref{lem:repeatable}}

We would like to analyze the sequence of outcomes for the state after a few measurement--recovery steps. Denote the CP maps for each $i$
\begin{align}
    \vK_i [\cdot ]\vK^{\dagger}_i&=: \CK_i[\cdot]\\
    \CM\CK_{i}[\cdot]&=: \CT_{i}[\cdot].
\end{align}

When we sequentially apply measurement-recovery, the sequence of measurement outcomes is nearly a product distribution.

\begin{lem}[Effectively independent samples]
\label{lem:TV}
    Consider the distribution over outcome sequences
\begin{align}
Q(i_1,i_2,\cdots, i_r) := \tr[\CT_{i_r}\cdots \CT_{i_1} [\vsigma]]
\end{align}
and the product measure
\begin{align}
P(i_1,i_2,\cdots, i_r) :=  \tr[\CT_{i_r} [\vsigma]]\cdots \tr[\CT_{i_1} [\vsigma]].
\end{align}
Suppose that $\sum_i \lnorm{\CT_i[\vsigma]-\vsigma\tr[\CT_i[\vsigma]]}_1 \le \epsilon.$
Then $\labs{P - Q}_{TV}\le r\epsilon.$
\end{lem}

\begin{proof}
We consider a telescoping sum
\begin{align}
        \sum_{i_r,\cdots ,i_1} \lnorm{\CT_{i_r}\cdots \CT_{i_1} [\vsigma]  - \vsigma \cdot \tr[\CT_{i_r}[\vsigma]] \cdots\tr[\CT_{i_1}[\vsigma]] }_1
    &\le \sum_{i_r,\cdots ,i_1} \lnorm{\CT_{i_r}\cdots \L( \CT_{i_1} [\vsigma] - \vsigma \tr[\CT_{i_1}\vsigma]\R)}_1\\
    &+ \sum_{i_r,\cdots ,i_1} \lnorm{\CT_{i_r}\cdots \L( \CT_{i_2} [\vsigma] - \vsigma \tr[\CT_{i_2}\vsigma]\R) \cdot\tr[\CT_{i_1}\vsigma] }_1\\
    &+\cdots\\
    &+ \sum_{i_r,\cdots ,i_1} \lnorm{\L( \CT_{i_r} [\vsigma] - \vsigma \tr[\CT_{i_r}\vsigma]\R) \cdots \tr[\CT_{i_1}\vsigma] }_1.
\end{align}
Pull the scalar factors of $\tr[\CT_{i_1}\vsigma]$ out of the norm, which sums to unity
\begin{align}
    (cont.) &\le \sum_{i_r,\cdots ,i_1} \lnorm{\CT_{i_r}\cdots \L( \CT_{i_1} [\vsigma] - \vsigma \tr[\CT_{i_1}\vsigma]\R)}_1\\
    &\ + \sum_{i_r,\cdots i_2} \lnorm{\CT_{i_r}\cdots \L( \CT_{i_2} [\vsigma] - \vsigma \tr[\CT_{i_2}\vsigma]\R)  }_1\\
    &\ +\cdots\\
    &\ + \sum_{i_r} \lnorm{ \CT_{i_r} [\vsigma] - \vsigma \tr[\CT_{i_r}\vsigma] }_1\\
    &\le  \sum_{i_r}\lnorm{ \CT_{i_r} [\vsigma] -  \vsigma \tr[\CT_{i_r}\vsigma] }_1 + \cdots + \sum_{i_1} \lnorm{ \CT_{i_1} [\vsigma] - \vsigma \tr[\CT_{i_1}\vsigma] }_1\tag*{(By~\autoref{lem:sum_trace} and sum over products of $\CT_{i_1}$)}\\
    &\le r \epsilon,
\end{align}
as advertised.
\end{proof}

In the above, we used an elementary inequality.

\begin{lem}\label{lem:sum_trace}
For any set of $CP$ maps $\{\CT_{i}\}$ such that $\sum_{i} \CT_i^{\dagger}[\vI] = \vI,$ we have that
\begin{align}
    \sum_{i} \lnorm{\CT_{i} [\vX]
     }_1 \le \lnorm{\vX}_1.
\end{align}
\end{lem}
\begin{proof}
    Consider the auxiliary classical-quantum channel (CPTP maps), which includes classical registers
    \begin{align}
    \CT': \vX\rightarrow \sum_i \CT_i[\vX] \otimes \ket{i}\bra{i}.
    \end{align}
    We observe that
    \begin{align}
        \sum_i \norm{\CT_{i}[\vX]}_1= \norm{\CT'[\vX]}_1\le \norm{\vX}_1
    \end{align}
    using the 1-norm contractivity of CPTP maps, which concludes the proof.
\end{proof}

We can now prove the main guarantee for the repeated measurement-recovery sequence.
\begin{proof}[Proof of~\autoref{lem:repeatable}]
    To estimate the outcome probability $\tr[\CK_i[\vsigma]],$ consider the Boolean variable $b_{\ell}$ associated with the $\ell-$th outcome and the empirical mean
\begin{align}
    b_{\ell}:=\indicator(i_{\ell}=i), \quad \hat{\mu}:= \frac{1}{r}\sum_{\ell=1}^r b_{\ell}.
\end{align}
    We will start with the idealized case of a perfect strong Markov property, with a perfect product distribution $P$ for the consecutive outcomes. Here, the expected empirical mean is exactly the desired outcome probability
\begin{align}
    \tr[\CK_i[\vsigma]] &= \frac{1}{r}\BE_P \sum_{\ell=1}^r b_{\ell}.
\end{align}
In the ideal case, by Hoeffding's,
\begin{align}
     \Pr_P\L(\frac{1}{r}\labs{\sum_{\ell=1}^r b_{\ell} - \BE_P \sum_{\ell=1}^r b_{\ell}}\ge \tau\R) \le 2\exp\L(-2r\tau^2\R).
\end{align}
We then compare back to the original nonproduct distribution $Q$ by the total variation bound (\autoref{lem:TV}) and the fact that
\begin{align}
     \Pr_Q\L(\frac{1}{r}\labs{\sum_{\ell=1}^r b_{\ell} - \BE_P \sum_{\ell=1}^r b_{\ell}}\ge \tau\R) \le 2\exp\L(-2r\tau^2\R) + r \epsilon.
\end{align}
Finally, we balance the two terms by
$r = \lceil \frac{1}{2\tau^2}\log ( \frac{4\tau^2}{\epsilon})\rceil$
\begin{align}
\delta \le \frac{\epsilon}{2\tau^2} +\epsilon \lceil \frac{1}{2\tau^2}\log ( \frac{4\tau^2}{\epsilon})\rceil \le \frac{\epsilon}{4\tau^2} + \frac{\epsilon}{2\tau^2}\log ( \frac{4\tau^2}{\epsilon})+\epsilon = \frac{\epsilon}{2\tau^2}\log ( \frac{4e\tau^2}{\epsilon})+\epsilon
\end{align}
to obtain the advertised failure probability.

\end{proof}

\subsection{Locally very close or locally very far}
The measurement--recovery property of strongly Markov states allows one to obtain the entire marginal with very high precision, thereby instantiating a single-shot statistical distinguisher between the two.

\begin{lem}[Locally close but distinguishable]
\label{lem:local_close_full}
Consider any two $\epsilon$-strongly Markov states $\vsigma_1$ and $\vsigma_2$ under the same recovery map $\CR$. Suppose that the marginal on $A$ is $\delta$ far in trace distance $\norm{\tr_{BC}[\vsigma_1]-\tr_{BC}[\vsigma_2]}_1 \ge \delta$.
Then there is a measurement-recovery procedure that distinguishes the case $\vsigma_1$ from $\vsigma_2$ with failure probability $\frac{16\epsilon}{\delta^2}\log ( \frac{e\delta^2}{8\epsilon})+2\epsilon.$
\end{lem}

\begin{proof}
The difference in 1-norm implies the existence of an optimal distinguishing operator $\vO_A = \vP_{A}-\vQ_A$ as the difference of two disjoint projectors such that $\tr[\vP_A\vsigma_1]-\tr[\vP_A\vsigma_2] =\delta/2$. Consider the measurement channel $\CK[\cdot]= \vP_{A}\cdot \vP_{A} + \vQ_{A}\cdot \vQ_{A}$; since the Kraus rank is two, $\vsigma_1$ and $\vsigma_2$ are both $\epsilon'$-strongly Markov for the measurement channel $\CK$, for $\epsilon' = 2\epsilon.$ By~\autoref{lem:repeatable} at $\tau = \delta/4$, we can measure $\vP_A$ to precision $\delta/4$ using the measurement--recovery protocol, thereby distinguishing $\vsigma_1$ and $\vsigma_2$ under the guaranteed failure probability.
\end{proof}

\begin{cor}
    In the setting of~\autoref{lem:local_close_full}, there exists a quasi-neighborhood such that the two marginals are nearly distinguishable.
\end{cor}

\subsection{Strongly Markov states are locally extremal}
The set of strongly Markov states has local marginals that exhibit extremal properties. Here, the locality of the recovery map is less crucial except that the \textit{same} recovery map must apply for both states $\vsigma_1$ and $\vsigma_2.$

\begin{lem}[Locally extremal]
    Consider a region $A$ and a recovery map $\CM$ for which an $\epsilon$-strongly Markov state $\vsigma$ is a mixture over two $\epsilon$-strongly Markov states.
    \begin{align}
    \vsigma = p_1\vsigma_1 + p_2 \vsigma_2.
\end{align}
    Then the local marginals are close
\begin{align}
\normp{\vsigma^{A}_1-\vsigma_2^{A}}{1} \le 2\sqrt{\frac{2\epsilon}{p_1p_2}}.
\end{align}

\end{lem}
 The contrapositive states that, if the marginals are sufficiently different, then the mixture would be far from strongly Markov, as a form of extremality. Since approximate stationarity is preserved under mixture, there must be a nontrivial correlation introduced by taking a mixture, which breaks the strong Markov property.

\begin{proof}
Denote
\begin{align}
     \norm{\vX-\vY}_1 \le \epsilon \quad \text{as}\quad  \vX \stackrel{\epsilon}{\approx} \vY.
\end{align}
    By assumption, for any $\vK$, $\norm{\vK}\le1,$
\begin{align}
\CR[\vK\vsigma_1\vK^{\dagger}]&\stackrel{\epsilon}{\approx} \vsigma_1\cdot \undersetbrace{=:q_1}{\tr[\vK \vsigma_1\vK^{\dagger}]}\\
    \CR[\vK\vsigma_2\vK^{\dagger}]&\stackrel{\epsilon}{\approx} \vsigma_2\cdot \undersetbrace{=:q_2}{\tr[\vK \vsigma_2\vK^{\dagger}]}.
\end{align}
Moreover,
\begin{align}
 \CR[\vK (p_1\vsigma_1 + p_2\vsigma_2)\vK^{\dagger}]&\stackrel{\epsilon}{\approx} \undersetbrace{=p_1q_1+p_2q_2}{\tr[\vK (p_1\vsigma_1 + p_2\vsigma_2)\vK^{\dagger}]} (p_1\vsigma_1+p_2\vsigma_2)\\
 \implies p_1q_1\vsigma_1 + p_2 q_2\vsigma_2 &\stackrel{\epsilon}{\approx} (p_1q_1+p_2q_2)(p_1\vsigma_1+p_2\vsigma_2).
\end{align}

Rearrange to obtain
\begin{align}
    p_1 p_2(q_1 - q_2) (\vsigma_1 - \vsigma_2) = p_1 (q_1 - p_1q_1 - p_2q_2)\vsigma_1 + p_2(q_2 - p_1q_1-q_2p_2) \vsigma_2 \stackrel{2\epsilon}{\approx} 0.
\end{align}
Taking traces with $\vK^{\dagger}\vK$ and using $\tr[\vK^{\dagger}\vK\vsigma_1]=q_1$ and $\tr[\vK^{\dagger}\vK\vsigma_2]=q_2$, we obtain
\begin{align}
    p_1 p_2(q_1 - q_2)^2 \le 2\epsilon.
\end{align}

By the variational characterization of the Schatten 1-norm,
\begin{align}
\norm{\vsigma_1^A-\vsigma_2^A}_1
=
2 \sup_{0 \le \vK^{\dagger}\vK \le I}
\labs{\tr[\vK^{\dagger}\vK(\vsigma_1^{A}-\vsigma_2^{A})]}
\le
2\sqrt{\frac{2\epsilon}{p_1 p_2}},
\end{align}
which concludes the proof.
\end{proof}

\bibliographystyle{alphaUrlePrint.bst}
\bibliography{ref}

\begin{thebibliography}{CKBG25}

\bibitem[BCV25]{bergamaschi2025structural}
Thiago Bergamaschi, Chi-Fang Chen, and Umesh Vazirani.
\newblock A structural theory of quantum metastability: Markov properties and area laws.
\newblock {\em arXiv preprint arXiv:2510.08538}, 2025.

\bibitem[BLMT25]{bakshi2025dobrushin}
Ainesh Bakshi, Allen Liu, Ankur Moitra, and Ewin Tang.
\newblock A dobrushin condition for quantum markov chains: Rapid mixing and conditional mutual information at high temperature.
\newblock {\em arXiv preprint arXiv:2510.08542}, 2025.

\bibitem[CG26]{chen2026efficient}
Chi-Fang Chen and Andr{\'a}s Gily{\'e}n.
\newblock Efficient shadow tomography of thermal states.
\newblock {\em arXiv preprint arXiv:2603.16845}, 2026.

\bibitem[CK25]{chen2025catalytic}
Chi-Fang Chen and Robbie King.
\newblock Catalytic tomography of ground states.
\newblock {\em arXiv preprint arXiv:2512.10247}, 2025.

\bibitem[CKBG23]{chen2023quantum}
Chi-Fang Chen, Michael~J Kastoryano, Fernando~GSL Brand{\~a}o, and Andr{\'a}s Gily{\'e}n.
\newblock Quantum thermal state preparation.
\newblock {\em arXiv preprint arXiv:2303.18224}, 2023.

\bibitem[CKBG25]{chen2025efficient}
Chi-Fang Chen, Michael Kastoryano, Fernando~GSL Brand{\~a}o, and Andr{\'a}s Gily{\'e}n.
\newblock Efficient quantum thermal simulation.
\newblock {\em Nature}, 646(8085):561--566, 2025.

\bibitem[CKG23]{chen2023efficient}
Chi-Fang Chen, Michael~J Kastoryano, and Andr{\'a}s Gily{\'e}n.
\newblock An efficient and exact noncommutative quantum gibbs sampler.
\newblock {\em arXiv preprint arXiv:2311.09207}, 2023.

\bibitem[CR25]{chen2025GibbsMarkov}
Chi-Fang Chen and Cambyse Rouz{\'e}.
\newblock Quantum gibbs states are locally markovian.
\newblock {\em arXiv preprint arXiv:2504.02208}, 2025.

\bibitem[FR15]{fawzi2015quantum}
Omar Fawzi and Renato Renner.
\newblock Quantum conditional mutual information and approximate markov chains.
\newblock {\em Communications in Mathematical Physics}, 340(2):575--611, 2015.

\bibitem[KK25]{kato2025clustering}
Kohtaro Kato and Tomotaka Kuwahara.
\newblock Clustering of conditional mutual information via quantum belief-propagation channels.
\newblock {\em arXiv preprint arXiv:2504.02235}, 2025.

\bibitem[Kuw24]{kuwahara2024clustering}
Tomotaka Kuwahara.
\newblock Clustering of conditional mutual information and quantum markov structure at arbitrary temperatures.
\newblock {\em arXiv preprint arXiv:2407.05835}, 2024.

\bibitem[SA25]{SA24}
Matteo Scandi and {\'A}lvaro~M. Alhambra.
\newblock Thermalization in open many-body systems and kms detailed balance, 2025, arXiv: {{\ttfamily 2505.20064}}.

\end{thebibliography}

\end{document}